\documentclass[11pt]{amsart}

\usepackage{witnessed}

\title{Witnessed k-Distance}

\author{Leonidas Guibas}
\email{guibas@cs.stanford.edu}
\address{Department of Computer Science, Stanford University, Stanford, CA}

\author{Quentin Merigot}
\email{quentin.merigot@imag.fr}
\address{Laboratoire Jean Kuntzmann, Université Grenoble I / CNRS}

\author{Dmitriy Morozov}
\email{dmitriy@mrzv.org}
\address{Departments of Computer Science and
  Mathematics, Stanford University, Stanford, CA}

\thanks{This work has been partly supported by ANR grant
  ANR-09-BLAN-0331-01, NSF grants FODAVA 0808515, CCF 1011228, and
  NSF/NIH grant 0900700.}

\begin{document}

\begin{abstract}
  Distance function to a compact set plays a central role in several
  areas of computational geometry. Methods that rely on it are robust to the
  perturbations of the data by the Hausdorff noise, but fail in the presence of
  outliers. The recently introduced \emph{distance to a measure}
  offers a solution by extending the distance
  function framework to reasoning about the geometry of probability measures,
  while maintaining theoretical guarantees
  about the quality of the inferred information.  A combinatorial
  explosion hinders working with distance to a measure as 
  an ordinary (power) distance function. In this paper, we analyze
  an approximation scheme that keeps the representation linear in the
  size of the input, while maintaining the guarantees on the inference
  quality close to those for the exact (but costly) representation.
\end{abstract}


\maketitle

%

\section{Introduction}

The problem of recovering the geometry and topology of compact sets
from finite point samples has seen several important developments in
the previous decade. Homeomorphic surface reconstruction algorithms
have been proposed to deal with surfaces in $\Rsp^3$ sampled without
noise \cite{amenta1999srv} and with moderate Hausdorff (local) noise
\cite{dey2006provable}. In the case of submanifolds of a higher
dimensional Euclidean space \cite{niyogi2008finding}, or even for more
general compact subsets \cite{chazal2006stc}, it is also possible, at
least in principle, to compute the homotopy type from a Hausdorff
sampling. If one is only interested in the homology of the underlying
space, the theory of persistent homology \cite{persistence-survey}
applied to Rips graphs provides an algorithmically tractable way to
estimate the Betti numbers from a finite Hausdorff sampling
\cite{chazal2008towards}.  

All of these constructions share a common feature: they estimate the
geometry of the underlying space by a union of balls of some radius
$r$ centered at the data points $P$. A different way to interpret this
union is as the $r$-sublevel set of the \emph{distance function} to
$P$, $\dist_P: x\mapsto \min_{p\in P} \norm{x - p}$. Distance
functions capture the geometry of their defining sets, and they are
stable to Hausdorff perturbations of those sets, making them
well-suited for reconstruction results.  However, they are also
extremely sensitive to the presence of outliers (i.e.~data points that
lie far from the underlying set); all reconstruction techniques that
rely on them fail even in presence of a single outlier.

To counter this problem, Chazal, Cohen-Steiner, and M\'e\-ri\-got
\cite{distance-to-measure} developed a notion of \emph{distance
  function to a probability measure} that retains the 
properties of the (usual) distance important for geometric 
inference. Instead of assuming an underlying compact set that is
sampled by the points, they assume an underlying probability measure
$\mu$ from which the point sample $P$ is drawn. The distance function
$\dd_{\mu,m_0}$ to the measure $\mu$ depends on a mass parameter $m_0
\in (0,1)$. This parameter acts as a smoothing term: a smaller $m_0$
captures the geometry of the support better, while a larger $m_0$
leads to better stability at the price of precision.  The crucial
feature of the function $\dd_{\mu,m_0}$ is its stability to the
perturbations of the measure $\mu$ under the Wasserstein distance,
defined in Section \ref{sec:wasserstein}. For
instance, the Wasserstein distance between the underlying measure
$\mu$ and the uniform probability measure on the point set $P$
can be small even if $P$ contains some outliers. When this happens,
the stability result ensures that distance function
$\dd_{\UnifMeasure_P, m_0}$ to the uniform probability measure
$\UnifMeasure_P$ on $P$ retains the geometric information
contained in the underlying measure $\mu$ and its support.

\paragraph{Computing with distance functions to measures} 
In this article we address the computational issues related to this new
notion. If $P$ is a subset of $\Rsp^d$ containing $N$ points, and $m_0
= k/N$, we will denote the distance function to the uniform measure on
$P$ by $\dd_{P,k}$. As observed in \cite{distance-to-measure}, the
value of $\dd_{P,k}$ at a given point $x$ is easy to compute: it is
the square root of the average squared distance from
the point $x$ to its $k$ nearest neighbors in $P$. However, most
inference methods require a way to represent the function,
or more precisely its sublevel sets, globally. It turns out that the
distance function $\dd_{P,k}$ can be rewritten as a minimum
\begin{equation}
\label{eq:dk:barycenters}
 \dd^2_{P,k}(x) = \min_{\bar{c}} \norm{x - \bar{c}}^2 - w_{\bar{c}},
\end{equation}
where $\bar{c}$ ranges over the set of barycenters of $k$ points in $P$
(see Section \ref{sec:kdistance}). Computational geometry provides a
rich toolbox to represent sublevel sets of such
functions, for example, via weighted $\alpha$-complexes \cite{alpha}.

The difficulty in applying these methods is that to get an equality in
\eqref{eq:dk:barycenters} the minimum number of barycenters to store is the
same as the number of order-$k$ Voronoi sites of $P$,
making this representation unusable even for modest input sizes. The
solution that we propose is to construct an approximation of the
distance function $\dd_{P,k}$, defined by the same equation as
\eqref{eq:dk:barycenters}, but with $\bar{c}$ ranging over a
smaller subset of barycenters. In this article, we study the quality
of approximation given by a \emph{linear-sized} subset: the
\emph{witnessed barycenters} defined as the barycenters of any $k$
points in $P$ whose order-$k$ Voronoi cell contains at least one of
the sample points. The algorithmic simplicity of the scheme is
appealing: we only have to find the $k-1$ nearest neighbors for each
input point. We denote by $\dw_{P,k}$ and call \emph{witnessed
  $k$-distance} the function defined by
Equation~\eqref{eq:dk:barycenters}, where $\bar{c}$ ranges over
the witnessed barycenters.

\paragraph{Contributions}
Our goal is to give conditions on the point cloud $P$
under which the witnessed $k$-distance $\dw_{P,k}$ provides a good
uniform approximation of the distance to measure $\dd_{P,k}$. We first
give a general multiplicative bound on the error produced by this
approximation. However, most of our paper (Sections
\ref{sec:witnessed-analysis} and \ref{sec:convergence}) analyzes the
uniform approximation error,
when $P$ is a set of independent samples from a measure
concentrated near a lower-dimensional subset of the Euclidean space. The
following is a prototypical example for our setting, although the
analysis we propose allows for a wider range of problems. Note that
some of the common settings in the literature either fit directly into
this example, or in its logic: the mixture of Gaussians
\cite{dasgupta1999learning} and off-manifold Gaussian noise in normal
directions \cite{niyogi2008topological} are two examples.

\begin{itemize}
\item[\Hyp1] We assume that the ``ground truth'' is an unknown
  probability measure $\mu$ whose \emph{dimension is bounded} by a
  constant $\ell \ll d$. Practically, this means that $\mu$ is
  concentrated on a compact set $K \subseteq \Rsp$ whose dimension is
  at most $\ell$, and that its mass distribution shouldn't ``forget''
  any part of $K$ (see Definition~\ref{def:dimension-bound}).  As an
  example $\mu$ could be the uniform measure on a smooth compact
  $\ell$-dimensional submanifold $K$, or on a finite union of such
  submanifolds.
\end{itemize}
This hypothesis ensures that the distance to the measure $\mu$ is
close to the distance to the support $K$ of $\mu$, and lets us recover
information about $K$. Our first result (Witnessed Bound Theorem \ref{th:witnessed-approx}) 
states that if the uniform measure to a point cloud $P$ is a good Wasserstein-approximation of $\mu$,
then the witnessed $k$-distance to $P$ provides a good approximation
of the distance to the underlying compact set $K$.
The bound we obtain is only a constant times worse than the bound for the exact
$k$-distance. 
\begin{itemize}
\item[\Hyp2] The second assumption is that we are not sampling
  directly from $\mu$, but through a noisy channel. We model this by
  considering that our measurements come from a measure $\nu$, which
  is obtained by adding noise to $\mu$. For instance, $\nu$ could be
  the result of the convolution of $\mu$ with a Gaussian distribution
  $\Gauss(0, d^{-1} \sigma^2 \id)$ whose variance is $\sigma^2$. More
  generally, $\nu$ can be any measure such that the
  \emph{Wasserstein distance} from $\mu$ to $\nu$ is at most
  $\sigma$. This generalization allows, in particular, to consider
  noise models that are not translation-invariant.
\item[\Hyp3] Finally, we suppose that our input data set $P \subseteq
  \Rsp^d$ consists of $N$ points drawn independently from the
  noisy measure $\nu$. Denote with $\UnifMeasure_P$ the uniform
  measure on $P$.
\end{itemize}
These two hypothesis allow us to control the Wasserstein distance between $\mu$ and
$\UnifMeasure_P$ with high probability.
We assume that the point cloud $P$ is gathered following the three
hypothesis above. Our second result states that the witnessed
$k$-distance to $P$ provides a good approximation of the distance to
the compact set $K$ with high probability, as soon as the amount of
noise $\sigma$ is low enough and the number of points $N$ is large
enough.

\begin{result}{Approximation Theorem (Theorem \ref{th:main})}
  Let $P$ be a set of $N$ points drawn according to the three
  hypothesis \Hyp1-\Hyp3, let $k \in \{1,\hdots, N\}$ and $m_0 =
  k/N$. Then, the error bound
$$ \norm{\dw_{P,k} - \dd_K}_\infty \leq 54 m_0^{-1/2} \sigma + 24
m_0^{1/\ell}\alpha_\mu^{-1/\ell}$$ holds with probability at least
$$1
- \gamma_\mu \exp(- \beta_\mu N \max(\sigma^{2+2\ell}, \sigma^4) -
\ell\ln(\sigma))$$ where the constants $\beta_\mu$ and $\gamma_\mu$
depend only on $\mu$.
\end{result}

We illustrate the utility of the bound with an example and a topological
inference statement in our final Section \ref{sec:discussion}.

\begin{figure}
 \subfigure[Data]{\includegraphics[width=.48\textwidth]{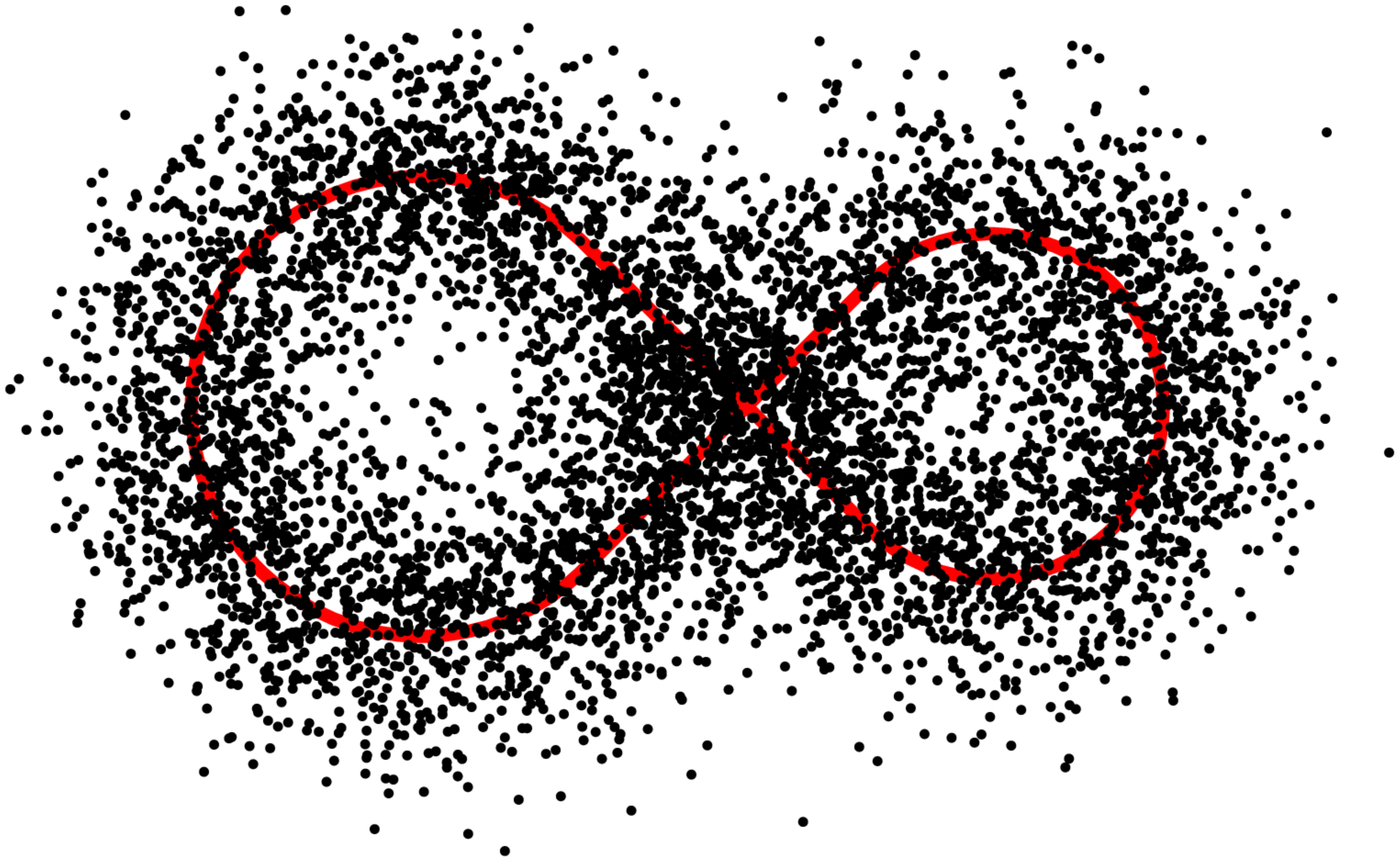} \label{fig:sample}}
 \subfigure[Sublevel sets]{\includegraphics[width=.48\textwidth]{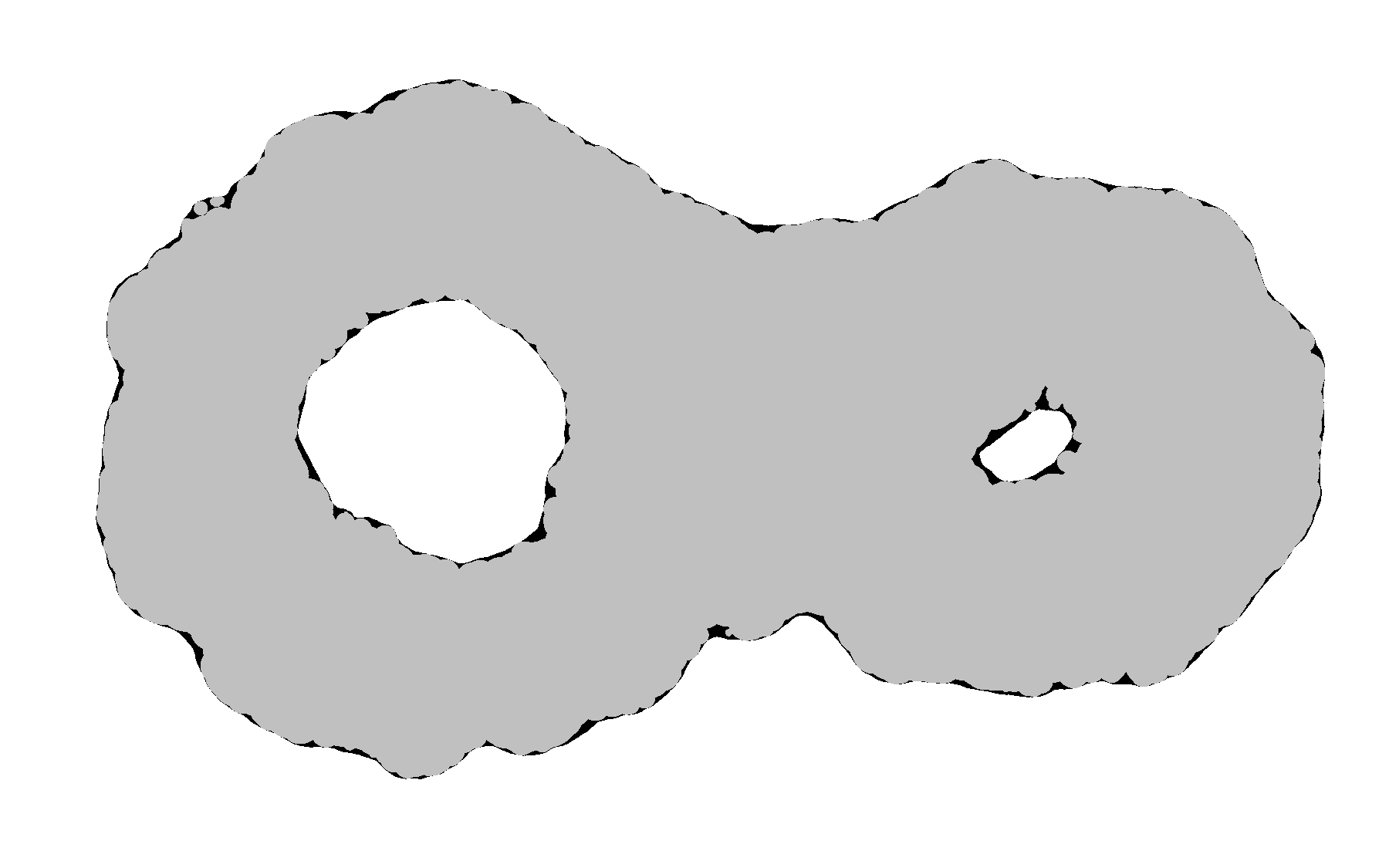} \label{fig:sublevelsets}}

 \caption{(a) 6000 points sampled from a sideways figure 8 (in red),
   with circle radii $R_1 = \sqrt{2}$ and $R_2 = \sqrt{9/8}$. The
   points are sampled from the uniform measure on the figure-8,
   convolved with the Gaussian distribution $\mathcal{N}(0,\sigma^2)$
   where $\sigma = .45$. (b) $r$-sublevel sets of the witnessed (in gray)
   and exact (additional points in black) $k$-distances with mass
   parameter $m_0 = 50/6000$, and $r = .239$.}
 \label{fig:figure-8}   
\end{figure}

\paragraph{Outline} The relevant background appears in Section
\ref{sec:background}. We present our approximation scheme together
with a general bound of its quality in Section \ref{sec:kdistance}. We
analyze its approximation quality for measures concentrated on
low-dimensional subsets of the Euclidean space in Section
\ref{sec:witnessed-analysis}. The convergence of the uniform measure
on a point cloud sampled from a measure of low complexity appears in
Section \ref{sec:convergence} and leads to our main result.

\section{Background}
\label{sec:background}

We begin by reviewing the relevant background.

\subsection{Measure}
Let us briefly recap the few concepts of measure theory that we use.
A \emph{non-negative measure} $\mu$ on the space $\Rsp^d$ is a mass
distribution. Mathematically, it is defined as a function that maps
every (Borel) subset $B$ of $\Rsp^d$ to a non-negative number
$\mu(B)$, which is \emph{additive} in the sense that
$\mu\left(\cup_{i\in \N} B_i\right) = \sum_i \mu(B_i)$ whenever
$(B_i)$ is a countable family of disjoint (Borel) subsets of
$\Rsp^d$. The \emph{total mass} of a measure $\mu$ is $\mass(\mu) =
\mu(\Rsp^d)$.  A measure $\mu$ is called a \emph{probability measure}
if its total mass is one.  The \emph{support} of a probability measure
$\mu$, denoted by $\Supp(\mu)$ is the smallest closed set whose
complement has zero measure. The \emph{expectation} or \emph{mean} of
$\mu$ is the point $\Exp(\mu) = \int_{\Rsp^d} x \dd\mu(x)$; the
variance of $\mu$ is the number $\sigma_\mu^2 = \int_{\Rsp^d} \norm{x
  - \Exp(\mu)}^2 \dd \mu(x)$.

Although the results we present are often more general, the typical
probability measures we have in mind are of two kinds: (i) the uniform
probability measure defined by the volume form of a lower-dimensional
submanifold of the ambient space and (ii) discrete probability
measures that are obtained through noisy sampling of probability
measures of the previous kind. For any finite set $P$ with $N$ points,
denote by $\UnifMeasure_P$ the uniform measure supported on $P$, i.e. the sum
of Dirac masses centered at $p\in P$ with weight $1/N$.

\subsection{Wasserstein distance}
\label{sec:wasserstein} A natural way to quantify the distance between
two measures is the \emph{Wasserstein distance}.  This distance
measures the $\LL^2$-cost of transporting the mass of the first
measure onto the second one.  A general study of this notion and its
relation to the problem of optimal transport appear in
\cite{villani2003tot}. We first give the general definition and then
explain its interpretation when one of the two measures has finite
support.

A \emph{transport plan} between two measures $\mu$ and $\nu$ with the
same total mass is a measure $\pi$ on the product space $\Rsp^d\times
\Rsp^d$ such that for every subsets $A,B$ of $\Rsp^d$, $\pi(A \times
\Rsp^d) = \mu(A)$ and $\pi(\Rsp^d \times B) = \nu(B)$. Intuitively,
$\pi(A\times B)$ represents the amount of mass of $\mu$ contained in
$A$ that will be transported to $B$ by $\pi$. The \emph{cost} of this
transport plan is given by
$$ c(\pi) := \left(\int_{\Rsp^d \times \Rsp^d} \norm{x - y}^2
\dd\pi(x,y) \right)^{1/2}$$ 
Finally, the \emph{Wasserstein distance} between $\mu$
and $\nu$ is the minimum cost of a transport plan between these
measures.
\label{def:wasserstein}

Consider the special case where the measure $\nu$ is supported on a
finite set $P$. This means that $\nu$ can be written as $\sum_{p\in P}
\alpha_p \delta_{p}$, where $\delta_p$ is the unit Dirac mass at
$P$. Moreover,  $\sum_p \alpha_p$ must equal the total
mass of $\mu$. A transport plan $\pi$ between $\mu$ and $\nu$
corresponds to a decomposition of $\mu$ into a sum of positive
measures $\sum_{p\in P} \mu_p$ such that $\mass(\mu_p) = \alpha_p$.
The squared cost of the plan defined by this decomposition is then
$$c(\pi) = \left(\sum_{p\in P} \left[
\int_{\Rsp^d} \norm{x - p}^2 \dd \mu_p(x)\right]\right)^{1/2}.$$

\paragraph{Wasserstein noise} Two properties of the Wasserstein
distances are worth mentioning for our purpose. Together, they show
that the Wasserstein noise and sampling model generalize the commonly
used empirical sampling with Gaussian noise model:

\begin{itemize}
\item Consider a probability measure $\mu$ and $f: \Rsp^d \to \Rsp$
  the density of a probability distribution centered at the origin,
  and denote by $\nu$ the result of the convolution of $\mu$ by
  $f$. Then, the Wasserstein distance between $\mu$ and $\nu$ is at
  most $\sigma$, where $\sigma^2 := \int_{\Rsp^d} \norm{x}^2 f(x) \dd
  x$ is the variance of the probability distribution defined by $f$.
\item Let $P$ denote a set of $N$ points drawn independently from a
  given measure $\nu$. Then, the the Wasserstein distance
  $\Wass_2(\nu, \UnifMeasure_P)$ between $\nu$ and the uniform
  probability measure on $P$ converges to zero as $N$ grows to
  infinity with high probability.  Examples of such asymptotic
  convergence results are common in statistics, e.g.
  \cite{bolley2007qci} and references therein.  In Proposition
  \ref{prop:empirical} below, we give a quantitative non-asymptotic
  result assuming that $\nu$ is low-dimensional \Hyp1.
\end{itemize}
Using the notation introduced in the two items above, one has
$$\lim\sup_{N \to +\infty} \Wass_2(\mu, \UnifMeasure_p) \leq \sigma$$
with high probability as the number of point grows to infinity. A more
quantitative version of this statement can be found in
Corollary~\ref{cor:empirical}.

\subsection{Distance-to-measure and $k$-distance}
In \cite{distance-to-measure}, the authors introduce a
distance to a probability measure as a way to infer the geometry and
topology of this measure in the same way the geometry and topology of
a set is inferred from its distance function. Given a probability
measure $\mu$ and a \emph{mass parameter} $m_0 \in (0,1)$, they define
a distance function $\dist_{\mu,m_0}$ which captures the properties of
the usual distance function to a compact set that are used for
geometric inference.  

\begin{definition} For any point $x$ in $\Rsp^d$, let
  $\delta_{\mu,m}(x)$ be the radius of the smallest ball centered at
  $x$ that contains a mass at least $m$ of the measure $\mu$. 
  The \emph{distance to the measure} $\mu$ with parameter $m_0$ is
  defined by $\dist_{\mu,m_0}(x) = m_0^{-1/2} \left(\int_{m=0}^{m_0}
    \delta_{\mu,m}(x)^2 \dd m \right)^{1/2}$.
\end{definition}

Given a point cloud $P$ containing $N$ points, the measure of interest
is the uniform measure $\UnifMeasure_P$ on $P$. When $m_0$ is a
fraction $k/N$ of the number of points (where $k$ is an integer), we
call \emph{$k$-distance} and denote by $\dist_{P,k}$ the distance to
the measure $\dd_{\UnifMeasure_P,m_0}$. The value of $\dist_{P,k}$ at
a query point $x$ is given by
$$\dist_{P,k}^2(x) = \frac{1}{k}\sum\limits_{p \in \nn_P^k(x)} \norm{x - p}^2.$$ 
where $\nn_P^k(x) \subseteq P$ denotes the $k$ nearest
neighbors in $P$ to the point $x \in \Rsp^d$. (Note that while the
$k$-th nearest neighbor itself might be ambiguous, on the boundary of
an order-$k$ Voronoi cell, the distance to the $k$-th nearest neighbor is
always well defined, and so is $\dist_{P,k}$.)

The most important property of the distance function $\dist_{\mu,m_0}$
is its stability, for a fixed $m_0$, under perturbations of the
underlying measure $\mu$.  This property provides a bridge between the
underlying (continuous) $\mu$ and the discrete measures
$\UnifMeasure_P$. According to \cite[Theorem
3.5]{distance-to-measure}, for any two probability measures $\mu$ and
$\nu$ on $\Rsp^d$,
\begin{equation}
  \norm{\dist_{\mu, m_0} - \dist_{\nu, m_0}}_\infty  \leq 
  m_0^{-1/2}\Wass_2(\mu,\nu),
\label{eq:inference}
\end{equation}
where $\Wass_2(\mu,\nu)$ denotes the Wasserstein distance between the
two measures. The bound in this inequality depends on the choice of
$m_0$, which acts as a smoothing parameter.

\section{Witnessed $k$-Distance}
\label{sec:kdistance}

In this section, we describe a simple scheme for approximating
the distance to a uniform measure, together with a general error bound. The
main contribution of our work, presented in Section
\ref{sec:witnessed-analysis}, is the analysis of the quality
of approximation given by this scheme when the input points come
from a measure concentrated on a lower-dimensional subset of the Euclidean
space.

\subdivision{$k$-Distance as a Power Distance} Given a set of points
$U = \{u_1,\hdots, u_n \}$ in $\Rsp^d$ with weights $w_u$ for every $u
\in U$, we call \emph{power distance} to $U$ the function $\Pow_{U}$
obtained as the lower envelope of all the functions $x \mapsto \norm{u
  - x}^2 - w_u$, where $u$ ranges over $U$. By Proposition~3.1 in
\cite{distance-to-measure}, we can express the square of any distance
to a measure as a power distance with non-positive weights. The
following proposition recalls this property of the $k$-distance
$\dd_{P,k}$.

\begin{proposition} For any $P \subseteq \Rsp^d$, denote by
  $\Bary^k(P)$ the set of barycenters of any subset of $k$ points
  in $P$. Then
\begin{equation}
 \dd^2_{P,k} = \min \left\{ \norm{x - \bar{c}}^2  - w_{\bar{c}};~ \bar{c} \in \Bary^k(P) \right\}, 
\label{eq:kdistance-power}
\end{equation}
where the weight of a barycenter $\bar{c} = \frac{1}{k} \sum_{i} p_i$
is given by $w_{\bar{c}} := -\frac{1}{k} \sum_{i} \norm{\bar{c}
  - p}^2$.
\label{prop:power}
\end{proposition}

\begin{proof}
  For any subset $C$ of $k$ points in $P$, define $$\delta_C^2(x) :=
  \frac{1}{k} \sum_{p \in C} \norm{x - p}^2$$  Denoting by $\bar{c}$
  the barycenter of the points in $C$, an easy computation
  shows
\begin{equation*}
  \delta_{C}^2(x)   
         = \frac{1}{k} \sum_{p \in C} \norm{x - p}^2 
         = \norm{x - \bar{c}}^2 - w_{\bar{c}}
\end{equation*}
where the weight is given by $w_{\bar{c}} = -\frac{1}{k}
\sum_{p \in C} \norm{\bar{c} - p}^2$. The proposition follows from the
definition of the $k$-distance.
\end{proof}

In other words, the square of the $k$-distance function to $P$
coincides exactly with the power distance to the set of barycenters
$\Bary^k(P)$ with the weights defined above. From this expression, it
follows that the sublevel sets of the $k$-distance $\dd_{P,k}$ are
finite unions of balls,
$$
\dist_{P,k}^{-1}([0, \rho])
= \bigcup\limits_{c \in \nn_P^k(\Rsp^d)} \B(\bar{c}, (\rho^2 + w_{\bar{c}})^{1/2}).
$$
Therefore, ignoring the complexity issues, it is possible to compute
the homotopy type of this sublevel set by considering the weighted
alpha-shape of $\Bary^k(P)$ (introduced in \cite{alpha}), which is a
subcomplex of the regular triangulation of the set of weighted
barycenters.

From the proof of Proposition~\ref{prop:power}, we also see that the
only barycenters that actually play a role in
\eqref{eq:kdistance-power} are the barycenters of $k$ points of $P$
whose order-$k$ Voronoi cell is not empty.  However, the dependence on
the number of non-empty order-$k$ Voronoi cells makes computation
intractable even for moderately sized point clouds in the Euclidean
space.

One way to avoid this difficulty is to replace the $k$-distance to $P$
by an approximate $k$-distance, defined as in
Equation~\eqref{eq:kdistance-power}, but where the minimum is taken
over a smaller set of barycenters. The question is then: given a point
set $P$, can we replace the set of barycenters $\Bary^k_P$ in the
definition of $k$-distance by a small subset $B$ while controlling the
approximation error $\norm{\Pow^{1/2}_{B} - \dd_{P,k}}_\infty$?

This approach is especially attractive since many geometric
and topological inference methods using distance functions to compact
sets or to measures continue to hold when one of
the distance functions is replaced by a good approximation \emph{in the
  class of power distances.}


\subdivision{Approximating by witnessed $k$-distance} In order to
approach this question, we consider a subset of the supporting
barycenters suggested by the input data which we call witnessed
barycenters. The answer to the question is then essentially positive
when the input point cloud $P$ satisfies the hypotheses \Hyp1-\Hyp3.

\begin{definition}
  For every point $x$ in $P$, the barycenter of $x$ and its $(k-1)$
  nearest neighbors in $P$ is called a \emph{witnessed
    $k$-barycenter}. Let $\Baryw^k(P)$ be the set of all such barycenters.
  We get one witnessed barycenter for every point $x$
  of the sampled point set, and define the \emph{witnessed
    $k$-distance},
\begin{equation*}
  \dw_{P,k} = \min \{ \norm{x - \bar{c}}^2 - w_{\bar{c}}; \bar{c} \in \Baryw^k(P) \}.
\end{equation*}
\end{definition}

Computing the set of all witnessed barycenters of a point set $P$ only
requires finding the $k-1$ nearest neighbors of every point in
$P$. This search problem has a long history in computational geometry
\cite{AM05, clarkson2006nearest, indyk2004nearest}, and now has
several practical implementation. 

\paragraph{General error bound}
Because the distance functions we consider are defined by minima, and
$\Baryw^k(P)$ is a subset of $\Bary^k(P)$, the witnessed $k$-distance
is always greater than the exact $k$-distance.  In the lemma below, we
give a general multiplicative upper bound. This lemma does not assume
any specific property for the input point set~$P$. However, even such
a coarse bound can be used to estimate Betti numbers of sublevel sets
of $\dd_{P,k}$, using arguments similar to those in
\cite{chazal2008towards}.

\begin{lemma}[General Bound]
\label{lemma:boundsqrt2}
For any finite point set  $P \subseteq \Rsp^d$ and $0 < k < \abs{P}$, one has
$$\dist_{P,k} \leq \dw_{P,k} \leq (2 + \sqrt{2}) \dist_{P,k}$$
\end{lemma}

\begin{proof} 
Let $y \in \Rsp^d$ be a point, and $\bar{p}$ the barycenter associated to a cell
that contains $y$. This translates into
$\dist_{P,k}(y) = \dd_{\bar{p}}(y)$. In particular, $\norm{\bar{p} - y} \leq
\dist_{P,k}(y)$ and $\sqrt{-w_{\bar{p}}} \leq \dist_{P,k}(y)$.

Let us find a witnessed barycenter $\bar{q}$ that is close to
$\bar{p}$. We know that $\bar{p}$ is the barycenters of $k$ points
$x_1,\hdots, x_n$, and that $- w_{\bar{p}} = \frac{1}{k}\sum_{i=1}^k \norm{x_i -
  \bar{p}}^2$. Consequently, there should exist an $x_i$ such that
$\norm{x_i - \bar{p}} \leq \sqrt{-w_{\bar{p}}}$. Let $\bar{q}$ be the barycenter
witnessed by $x$.  Then,
\begin{align*}
\dw_{P,k}(y) \leq \dd_{\bar{q}}(y) 
&\leq  \dd_{\bar{q}}(x) + \norm{x - y} \\
&\leq \dd_{\bar{p}}(x) +  \norm{x - \bar{p}} + \norm {\bar{p} - y}
\end{align*}
Combining the inequality
$$\dd_{\bar{p}}(x) = \left(\norm{x - \bar{p}}^2 - w_{\bar{p}}\right)^{1/2} \leq \sqrt{2}
\sqrt{-w_{\bar{p}}}$$ together with $\norm{x - \bar{p}} \leq
\sqrt{-w_{\bar{p}}}$, we get
\begin{align*}
\dw_{P,k}(y) &\leq (1+\sqrt{2})  \sqrt{-w_{\bar{p}}} + \norm{\bar{p} - y}  \\
&\leq (2+\sqrt{2}) \dist_{P,k}(y) \qedhere
\end{align*}
\end{proof}

\section{Approximation Quality}
\label{sec:witnessed-analysis}

Let us recall briefly our hypothesis \Hyp1-\Hyp3. There is an ideal,
well-conditioned measure $\mu$ on $\Rsp^d$ supported on an unknown
compact set $K$.  We also have a noisy version of $\mu$, that is
another measure $\nu$ with $\Wass_2(\mu,\nu) \leq \sigma$, and we
suppose that our data set $P$ consists of $N$ points independently
sampled from $\nu$. In this section we give conditions under which the
witnessed $k$-distance to $P$ provides a good approximation of the
distance to the underlying set $K$.

\subdivision{Dimension of a measure} First, we make precise the main
assumption \Hyp1 on the underlying measure $\mu$, which we use to
bound the approximation error made when replacing the exact by the
witnessed $k$-distance.  We require $\mu$ to be low dimensional in the
following sense.


\begin{definition}
\label{def:dimension-bound}
A measure $\mu$ on $\Rsp^d$ is said to have \emph{dimension at most
  $\ell$}, which we denote by $\dim\mu \leq \ell$, if there is a
positive constant $\alpha_\mu$ such that the amount of mass contained
in the ball $B(p,r)$ is at least $\alpha_\mu r^\ell$, for every point
$p$ in the support of $\mu$ and every $r$ smaller than the diameter of
this support.
\end{definition}

The important assumption here is that the lower bound $\mu(\B(p,r))
\geq \alpha r^\ell$ should be true for some positive constant $\alpha$
and for $r$ smaller than a given constant $R$. The choice of $R =
\diam(\Supp(\mu))$ provides a normalization of the constant
$\alpha_\mu$ and slightly simplifies the statements of the results.

Let $M$ be an $\ell$-dimensional compact submanifold of $\Rsp^d$, and
$f: M \to \Rsp$ a positive weight function on $M$ with values 
bounded away from zero and infinity. Then, the dimension of the volume
measure on $M$ weighted by the function $f$ is at most $\ell$.  A
quantitative statement can be obtained using the Bishop-G\"unther
comparison theorem; the bound depends on the maximum absolute
sectional curvature of the manifold $M$ (see e.g. Proposition 4.9 in
\cite{distance-to-measure}).  Note that the positive lower
bound on the density is really necessary. For instance, the dimension of
the standard Gaussian distribution $\mathcal{N}(0,1)$ on the real line
is not bounded by $1$ --- nor by any positive constant.
(This fact follows since the density of this distribution
decreases to zero faster than any polynomial as one moves away from
the origin.)

It is easy to see that if $m$ measures $\mu_1,\hdots,\mu_m$ have
dimension at most $\ell$, then so does their sum. Consequently, if
$(M_j)$ is a finite family of compact submanifolds of $\Rsp^d$ with
dimensions $(d_j)$, and $\mu_j$ is the volume measure on $M_j$
weighted by a function bounded away from zero and infinity, the
dimension of the sum $\mu = \sum_{j=1}^m \mu_j$ is at most $\max_{j}
d_j$.

\subdivision{Bounds} In the remaining of this section, we bound the
error between the witnessed $k$-distance $\dw_{P,k}$ and the
(ordinary) distance $\dist_K$ to the compact set $K$. We start from a
proposition from \cite{distance-to-measure} that bounds the error
between the exact $k$-distance $\dd_{P,k}$ and $\dd_K$:

\begin{theorem}[Exact Bound]
  Let $\mu$ denote a probability mea\-su\-re with dimension at most
  $\ell$, and supported on a set. Consider the uniform measure
  $\UnifMeasure_P$ on a point cloud $P$, and set $m_0=
  k/\abs{P}$. Then
\label{prop:approxbydist}
$$\norm{\dd_{P,k} - \dd_K}_\infty \leq m_0^{-1/2} \Wass_2(\mu, \UnifMeasure_P) +
\alpha_\mu^{-1/\ell}m_0^{1/\ell}.$$
\end{theorem}

\begin{proof}
Recall that $\dd_{P,k} = \dd_{\UnifMeasure_P, m_0}$. Using the triangle
inequality and Equation~\eqref{eq:inference}, one has
\begin{align*}
\norm{\dd_{\UnifMeasure_P,m_0} - \dd_K}_\infty &\leq 
\norm{\dd_{\mu,m_0} - \dd_{\UnifMeasure_P,m_0}}_\infty + \norm{\dd_{\mu,m_0} - \dd_K}_\infty
\\
&\leq m_0^{-1/2} \Wass_2(\mu, \UnifMeasure_P) + \norm{\dd_{\mu,m_0} - \dd_K}_\infty
\end{align*}
Then, from Lemma 4.7 in \cite{distance-to-measure},
$\norm{\dd_{\mu,m_0} - \dd_K}_\infty \leq
\alpha_\mu^{-1/\ell}m_0^{1/\ell}$, and the claim follows. 
\end{proof}

In the main theorem of this section, the exact $k$-distance in the
above bound is replaced by the witnessed $k$-distance.

\begin{theorem}[Witnessed Bound]
\label{th:witnessed-approx}
Let $\mu$ be a probability measure satisfying the dimension assumption
and let $K$ be its support. Consider the uniform measure $\UnifMeasure_P$ on a
point cloud $P$, and set $m_0= k/\abs{P}$. Then,
$$ \norm{\dw_{P,k} - \dd_K}_\infty \leq 6 m_0^{-1/2} \Wass_2(\mu,\UnifMeasure_P) + 24
m_0^{1/\ell}\alpha_\mu^{-1/\ell}.$$
\end{theorem}

Observe that the error term given by this theorem is a constant factor
times the bound in the previous theorem.  Before
proceeding with the proof, we prove an auxiliary lemma, which
emphasizes that a measure $\nu$, close to a measure $\mu$ satisfying an upper 
dimension bound (as in Definition
\ref{def:dimension-bound}), remains concentrated around the support of
$\mu$.

\begin{lemma}[Concentration]
  Let $\mu$ be a probability measure satisfying the dimension
  assumption, and $\nu$ be another probability measure. Let $m_0$ be a
  mass parameter. Then, for every point $p$ in the support of $\mu$,
  $\nu(\B(p,\eta)) \geq m_0$, where $\eta = m_0^{-1/2} \Wass_2(\mu,
  \nu) + 4 m_0^{1/2+1/\ell}\alpha_\mu^{-1/\ell}$.
\label{lemma:size-ball}
\end{lemma}

\begin{proof}
  Let $\pi$ be an optimal transport plan between $\nu$ and $\mu$. For
  a fixed point $p$ in the support of $K$, let $r$ be the smallest
  radius such that $\B(p,r)$ contains at least $2m_0$ of mass $\mu$.
  Consider now a submeasure $\mu'$ of $\mu$ of mass exactly
  $2m_0$ and whose support is contained in the ball $\B(p,r)$. This
  measure is obtained by transporting a submeasure $\nu'$ of $\nu$ by
  the optimal transport plan $\pi$.  Our goal is to determine for what
  choice of $\eta$ the ball $\B(p,\eta)$ contains a $\nu'$-mass (and, 
  therefore, a $\nu$-mass) of at least $m_0$.  We make use of the
  Chebyshev's inequality for $\nu'$ to bound the mass of $\nu'$
  \emph{outside} of the ball $\B(p,\eta)$:
\begin{equation}
\begin{aligned}
\nu'(\Rsp^d \setminus \B(p,\eta)) &= \nu'(\{ x \in \Rsp^d;~ \norm{x -
  p} \geq \eta\})  \\
&\leq \frac{1}{\eta^2} \int \norm{x - p}^2 \dd\nu'
\end{aligned}
\label{eq:conc1}
\end{equation}
Observe that the right hand term of this inequality is exactly the
Wasserstein distance between $\mu'$ and the Dirac mass $2m_0
\delta_p$. We bound it using the triangle inequality for the
Wasserstein distance:
\begin{equation}
\begin{aligned}
\int \norm{x - p}^2 \dd \nu' &= \Wass_2^2(\nu', 2m_0 \delta_p) \\
&\leq (\Wass_2(\mu',\nu') + \Wass_2(\mu', 2 m_0 \delta_p))^2 \\
&\leq (\Wass_2(\mu,\nu) + 2 m_0 r)^2
\end{aligned}
\label{eq:conc2}
\end{equation}
Combining equations \eqref{eq:conc1} and \eqref{eq:conc2}, we get:
\begin{align*}
\nu(\bar{\B}(p,\eta)) \geq \nu'(\bar{\B}(p,\eta)) &\geq \nu'(\Rsp^d) - 
\nu'(\Rsp^d \setminus \B(p,\eta)) \\
&\geq 2 m_0 - \frac{(\Wass_2(\mu,\nu) + 2 m_0 r)^2}{\eta^2}.
\end{align*}
By the lower bound on the dimension of $\mu$, and the definition of
the radius $r$, one has $r \leq (2m_0/\alpha_{\mu})^{1/\ell}$. Hence,
the ball $\bar{\B}(p,\eta)$ contains a mass of at least $m_0$ as soon
as
$$ \frac{(\Wass_2(\mu,\nu)+\alpha_\mu^{-1} 2^{1+1/\ell}  m_0^{1+1/\ell})^2}{\eta^2}
\leq  m_0.$$
This will be true, in particular, if 
$\eta$ is larger than 
\begin{equation*}
\Wass_2(\mu,\nu) m_0^{-1/2}  + 4  \alpha_\mu^{-1/\ell} m_0^{1/2 + 1/\ell}.\qedhere
\end{equation*}
\end{proof}

\begin{proof}[Proof of the Witnessed Bound Theorem]
Since the witnessed $k$-distance is a minimum over fewer barycenters,
it is larger than the real $k$-distance. Using this fact and
the Exact Bound Theorem one gets the lower bound:

$$ \dw_{P,k}
\geq \dd_{P,k} \geq \dd_K - m_0^{-1/2} \Wass_2(\mu, \UnifMeasure_P) + 
\alpha_\mu^{-1/\ell} m_0^{1/\ell}$$

For the upper bound, if we set $\eta$ as in Lemma
\ref{lemma:size-ball}, for every point $p$ in $K$, the ball $\B(p,
\eta)$ contains at least $k$ points in $P$. Consider one of these
points $x_1$; its $(k-1)$ nearest neighbors $x_2, \hdots,x_{k}$
in $P$ cannot be at a distance greater than $2\eta$ from $x_1$.
Hence, the points $x_1,\hdots,x_k$ belong to the ball $\B(p, 3\eta)$
and so does their barycenter. This shows that the set $W$ of witnessed
barycenters, obtained by this construction, is a $3\eta$-covering of
$K$, that is $\dd_W \leq \dd_K + 3\eta$. Since the weight of any
barycenter in $W$ is at most $3\eta$, we get $\dw_{P,k} \leq
\dd_W + 3\eta$.  To sum up, $$\dw_{P,k} \leq \dd_W + 3\eta \leq \dd_K
+ 6\eta$$ Replacing $\eta$ by its value from the Concentration Lemma
concludes the proof.
\end{proof}

\section{Convergence under Empirical Sampling}
\label{sec:convergence}

One term remains moot in the bound in Theorem~\ref{th:witnessed-approx}, 
namely the Wasserstein distance $\Wass_2(\mu, \UnifMeasure_P)$.
In this section, we analyze its convergence. The rate
depends on the complexity of the measure $\mu$, defined below. The
moral of this section is that if a measure can be well approximated
with few points, then it is also well approximated by random sampling.

\begin{definition}
  The \emph{complexity} of a probability measure $\mu$ at a scale
  $\eps > 0$ is the minimum cardinality of a finitely supported
  probability measure $\nu$ which $\eps$-approximates $\mu$ in the
  Wasserstein sense, i.e. such that $\Wass_2(\mu,\nu) \leq \eps$. 
  We denote this number by $\LebNum_\mu(\eps)$.
\end{definition}

Observe that this notion is very close to the \emph{$\eps$-covering
  number} of a compact set $K$, denoted by $\LebNum_K(\eps)$, which
counts the minimum number of balls of radius $\eps$ needed to cover
$K$.  It's worth noting that if measures $\mu$ and $\nu$ are close ---
as are the measure $\mu$ and its noisy approximation $\nu$ in the
previous section --- and $\mu$ has low complexity, then so does the
measure $\nu$.  The following lemma shows that measures satisfying the
dimension assumption have low complexity. Its proof follows from a
classical covering argument, that can be found e.g. in Proposition~4.1
of \cite{kloeckner2010approximation}.

\begin{lemma}[Dimension-Complexity]
Let $K$ be the support of a measure $\mu$ with $\dim \mu \leq
\ell$. Then,
\begin{itemize}
\item[(i)] for every positive $\eps$, $\LebNum_K(\eps) \leq
  \alpha_{\mu}/\eps^\ell$. Said otherwise, the \emph{upper
    box-counting dimension} of $K$ is bounded:
$$ \dim(K) := \limsup_{\eps \rightarrow 0}
\log(\LebNum_K(\eps))/\log(1/\eps) \leq \ell.$$
\item[(ii)] for every positive $\eps$, $\LebNum_\mu(\eps) \leq
  \alpha_{\mu} 5^\ell/\eps^\ell$.
\end{itemize}
\label{lemma:complexity}
\end{lemma}

\begin{theorem}[Convergence]
\label{prop:empirical}
Let $\mu$ be a probability measure on $\Rsp^d$ whose support has
diameter at most $D$, and let $P$ be a set of $N$ points independently
drawn from the measure $\mu$. Then, $\eps>0$,
\begin{align*}
\Prob(\Wass_2(\UnifMeasure_P, \mu) \leq 4 \eps) \geq 1 &-
\LebNum_\mu(\eps)\exp(-2 N \eps^2/(D\LebNum_\mu(\eps))^2) \\
& - \exp(-2N
\eps^4/D^2)
\end{align*}
\end{theorem}

\begin{proof}
  Let $n$ be a fixed integer, and $\eps$ be the minimum Wasserstein
  distance between $\mu$ and a measure $\bar{\mu}$ supported on (at
  most) $n$ points. Let $S$ be the support of the optimal measure
  $\bar{\mu}$, so that $\bar{\mu}$ can be decomposed as $\sum_{s\in S}
  \alpha_s \delta_{s}$ ($\alpha_s \geq 0$).  Let $\pi$ be an optimal
  transport plan between $\mu$ and $\bar{\mu}$; this is equivalent to
  finding a decomposition of $\mu$ as a sum of $n$ non-negative
  measures $(\pi_s)_{s\in S}$ such that $\mass(\pi_s) = \alpha_s$, and
$$ \sum_{s\in S} \int \norm{x - s}^2 \dd \pi_s(x) = \eps^2 = \Wass_2(\mu,\bar{\mu})^2$$

Drawing a random point $X$ from the measure $\mu$ amounts to (i)
choosing a random point $s$ in the set $S$ (with probability $\alpha_s$)
and (ii) drawing a random point $X$ following the distribution $\pi_s$.
Given $N$ independent points $X_1,\hdots, X_N$ drawn from the measure
$\mu$, denote by $I_{s,N}$ the proportion of the $(X_i)$ for which the
point $s$ was selected in step (i). Hoeffding's inequality allows to
easily quantify how far the proportion $I_{s,N}$ deviates from
$\alpha_s$: $\Prob(\abs{I_{s,N} - \alpha_s} \geq \delta) \leq \exp(-2
N\delta^2)$.  Combining these inequalities for every point $s$ and using the union bound yields
$$\Prob\left(\sum_{s \in S} \abs{I_{s,N} - \alpha_s} \leq \delta\right) \geq
1 - n \exp(-2 N \delta^2/n^2).$$

For every point $s$, denote by $\tilde{\pi}_s$ the distribution of the
distances to $s$ in the submeasure $\pi_s$, i.e. the measure on the
real line defined by $ \tilde{\pi}_s(I) := \pi_s(\{ x \in \Rsp^d;
\norm{x - s} \in I \})$ for every interval $I$. Define $\tilde{\mu}$
as the sum of the $\tilde{\pi}_s$; by the change of variable formula
one has
$$\int_\Rsp t^2
\dd\tilde{\mu}(t) = \sum_s \int_\Rsp t^2 \dd\tilde{\pi}_s = \sum_s
\int_{\Rsp^d} \norm{x-s}^2 \dd\pi_s = \eps^2$$ Given a random point
$X_i$ sampled from $\mu$, denote by $Y_i$ Euclidean distance between
the point $X_i$ and the point $s$ chosen in step (i). By construction,
the distribution of $Y_i$ is given by the measure $\tilde{\mu}$; using
the Hoeffding inequality again one gets
$$ \Prob\left(\frac{1}{N} \sum_{i=1}^N Y^2_i \geq (\eps + \eta)^2\right) \leq 
1 - \exp(-2N\eta^2 \eps^2/D^2).$$

In order to conclude, we need to define a transport plan from the
empirical measure $\UnifMeasure_P = \frac{1}{N} \sum_{i=1}^N \delta_{X_i}$ to
the finite measure $\bar{\mu}$. To achieve this, we order the
points $(X_i)$ by increasing distance $Y_i$; then transport every
Dirac mass $\frac{1}{N} \delta_{X_i}$ to the corresponding point $s$
in $S$ until $s$ is ``full'', i.e. the mass $\alpha_s$ is reached.
The squared cost of this transport operation is at most $\frac{1}{N}
\sum_{i=1}^N Y^2_i$. Then distribute the remaining mass among the $s$
points in any way; the cost of this step is at most $D$ times
$\sum_{s\in S} \abs{I_{s,N} - \alpha_s}$.
The total cost of this transport plan is the sum of these two
costs. From what we have shown above, setting $\eta = \eps$ and
$\delta = \eps / D$, one gets
\begin{align*}
 \Prob(\Wass_2(\UnifMeasure_P, \mu) \leq 4 \eps) \geq 1 &- n \exp(-2 N
\eps^2/(Dn)^2) \\
& - \exp(-2N \eps^4/D^2) \qedhere
\end{align*}
\end{proof}

As a consequence of the Dimension-Complexity
Lemma~\ref{lemma:complexity} and of the Convergence
Theorem~\ref{prop:empirical}, any measure $\mu$ satisfying an upper
bound on its dimension is well approximated by empirical sampling.  A
result similar to the Convergence Theorem follows when the samples are
drawn not from the original measure $\mu$, but from a ``noisy''
approximation $\nu$ which need not be compactly supported:

\begin{corollary}[Noisy Convergence]
  Let $\mu,\nu$ be two probability measures on $\Rsp^d$ with
  $\Wass_2(\mu,\nu) = \sigma$, and $P$ be a set of $N$ points drawn
  independently from the measure $\nu$. Then,
\begin{align*}
\Prob(\Wass_2(\UnifMeasure_P, \mu) \leq 9\sigma) \geq 1 &-
  \LebNum_\mu(\sigma)\exp(-8 N \sigma^2/(D\LebNum_\mu(\sigma))^2) \\&-
  \exp(-32 N \sigma^4/D^2).
\end{align*}
\label{cor:empirical}
\end{corollary}

\begin{proof}
  One only needs to apply the previous Convergence
  Theorem  to the measures $\nu$ and $\UnifMeasure_P$:
\begin{align}
\label{eq:noisyconv}
\Prob(\Wass_2(\nu, \UnifMeasure_P) \leq 4 \eps) \geq 1 &-
\LebNum_\mu(\eps)\exp(-2 N \eps^2/(D\LebNum_\nu(\eps))^2)\notag\\& - \exp(-2N
\eps^4/D^2)
\end{align} Set $\eps = 2 \sigma$ and recall that by definition
$\LebNum_\nu(2\sigma) \leq \LebNum_\mu(\sigma)$. Then, using
$\Wass_2(\UnifMeasure_P, \mu) \leq \Wass_2(\UnifMeasure_P, \nu) +
\sigma$ one has
$$ \Prob(\Wass_2(\UnifMeasure_P, \mu) \leq 9 \sigma) \geq 
\Prob(\Wass_2(\UnifMeasure_P, \nu) \leq 8 \sigma)$$
We conclude by using Eq. \eqref{eq:noisyconv} with $\eps=2\sigma$. 
\end{proof}

It is now possible to combine Theorem~\ref{th:witnessed-approx}
(Witnessed Bound), Corollary~\ref{cor:empirical} (Noisy Convergence)
and Lemma~\ref{lemma:complexity} (Dimension-Complexity) to get the
following probabilistic statement.

\begin{theorem}[Approximation]
Suppose that $\mu$ is a measure satisfying the dimension assumption,
supported on a set $K$ of diameter $D$, and $\nu$ a noisy
approximation of $\mu$, i.e. $\Wass_2(\mu,\nu) \leq \sigma$.  Let $P$
be a set of $N$ points independently sampled from $\nu$. Then, the
inequality
$$ \norm{\dw_{P,k} - \dd_K}_\infty \leq 54 m_0^{-1/2} \sigma + 24
m_0^{1/\ell}\alpha_\mu^{-1/\ell} $$ holds with probability at least 
$$1 - \gamma_\mu \exp(- \beta_\mu N \max(\sigma^{2+2\ell},\sigma^4) -
\ell\ln(\sigma)),$$ where $\beta_\mu = \frac{1}{D^2}
\max\left[\frac{8}{(\alpha_\mu 5^\ell)^2}, 32\right]$ and $\gamma_\mu = 1
+ \alpha_\mu 5^\ell$.
\label{th:main}
\end{theorem}

\begin{proof} 
Thanks to the Witnessed Bound Theorem and the Noisy Convergence Corollary, 
the inequality holds with probability at least:
$$1 - \LebNum_\mu(\sigma)) \exp(-8 N
\sigma^2/(D\LebNum_\mu(\sigma))^2) - \exp(-32 N \sigma^4/D^2) $$We
use Lemma \ref{lemma:complexity} to lower bound the covering number
$\LebNum_\mu(\sigma)$ by $\alpha_\mu 5^\ell/\sigma^\ell$. Hence, the
previous expression is bounded from below by
\begin{align*}
1 &- \alpha_\mu 5^\ell \exp(-8 N
\sigma^{2+2\ell}/(D \alpha_\mu 5^\ell)^2 
- \ell \ln(\sigma)) - \exp(-32 N \sigma^4/D^2) \\
&\qquad\geq 1 - \gamma_\mu \exp(-\beta_\mu N \max(\sigma^{2+2\ell}, \sigma^4) - \ell\ln(\sigma)) 
\end{align*}
where $\gamma_\mu = 1 + \alpha_\mu 5^\ell$ and 
$\beta_\mu = \frac{1}{D^2} \max\left[\frac{8}{(\alpha_\mu 5^\ell)^2}, 32\right]$, 
as stated in the theorem.
\end{proof}

\section{Discussion}
\label{sec:discussion}

We illustrate the utility of the bound in the Witnessed Bound
Theorem by example and an inference 
statement. Figure~\ref{fig:figure-8} shows $6000$ points drawn from
the uniform distribution on a sideways figure-8 (in red), convolved
with a Gaussian distribution. The ordinary distance function to the
point set has no hope of recovering geometric information out of these
points since both loops of the figure-8 are filled in.  On the right,
we show the sublevel sets of the distance to the uniform measure on
the point set, both the witnessed $k$-distance and the exact
$k$-distance. Both functions recover the topology of figure-8, the
bits missing from the witnessed $k$-distance smooth out the boundary
of the sublevel set, but do not affect the image at large.

\paragraph{Inference}
Suppose that we are in the conditions of the Approximation Theorem, but
additionally we assume that the support $K$ of the original measure $\mu$ has
a \emph{weak feature size} larger than $R$. This means that the
distance function $\dd_K$ has no critical value in $[0,R]$, and
implies that all the offsets $K^r = \dd_{K}^{-1}[0,r]$ of $K$ are
homotopy equivalent for $r \in (0,R)$. Suppose again that we have
drawn a set $P$ of $N$ points from a Wasserstein approximation $\nu$
of $\mu$, such that $\Wass_2(\mu,\nu)\leq \sigma$. From the Approximation
Theorem, we have 
\begin{equation*}
 \norm{\dw_{P,k} - \dd_K}_\infty \leq e(m_0):= 54 m_0^{-1/2} \sigma +
 24 m_0^{1/\ell}\alpha_\mu^{-1/\ell}
\end{equation*}
with high probability as $N$ goes to infinity.  Then, the standard
argument \cite{cohensteiner2007spd} shows that the Betti numbers of
the compact set $K$ can be inferred from the function $\dw_{P,k}$,
which is defined only from the point sample $P$, as long as $e(m_0)$
is less than $R/4$ (see the Appendix). In the language of
persistent homology \cite{persistence-survey}, the persistent Betti
numbers $\Betti^{(e(m_0),3e(m_0))}$ of the function $\dw_{P,k}$ are
equal to the Betti numbers of the set $K$, $\Betti(K)$.

\begin{figure*}[t]
    \centering
    \includegraphics[width=\textwidth]{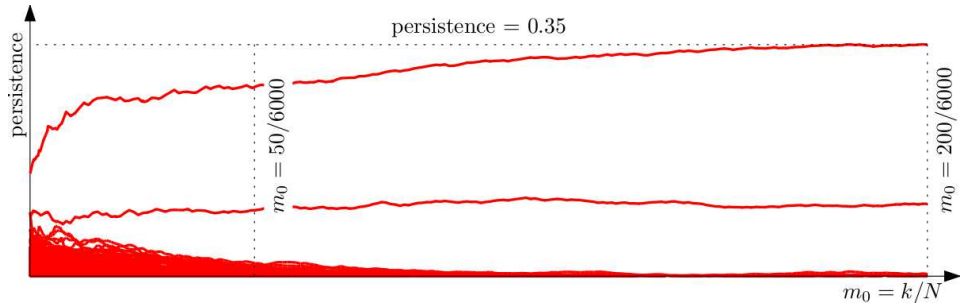}
    \caption{(PL-approximation of the) 1-dimensional persistence vineyard of 
             the witnessed $k$-distance function.
             Topological features of the space, obscured by noise
             for low values of $m_0$, stand out as we increase the mass
             parameter.}
    \label{fig:vineyard}
\end{figure*}

\paragraph{Choice of the mass parameter}
This language also suggests a strategy for choosing a mass parameter
$m_0$ for the distance to a measure, a question that has not been
addressed by the original paper \cite{distance-to-measure}. For every
mass parameter $m_0$, the $p$-dimensional \emph{persistence diagram}
$\Pers_p(\dd_{\mu,m_0})$ is a set of points $\{(b_i(m_0),
d_i(m_0))\}_i$ in the extended plane $(\Rsp \cup\{\infty\})^2$.  Each
of these points represents a homology class of dimension $p$ in the
sublevel sets of $\dd_{\mu,m_0}$; $b_i(m_0)$ and $d_i(m_0)$ are the
values at which it is born and dies. Since the distance to measure
$\dd_{\UnifMeasure_P,m_0}$ depends continuously on $m_0$, by
\cite{cohensteiner2007spd} so do its persistence diagrams. Thus, one
can use the algorithm in \cite{vineyards} to track their
evolution. Figure~\ref{fig:vineyard} illustrates such a construction
for the point set in Figure~\ref{fig:figure-8} and the witnessed
$k$-distance. It displays the evolution of the persistence $(d_1(m_0)
- b_1(m_0))$ of each of the 1-dimensional homology classes as $m_0$
varies, thus highlighting the choices of the mass parameter that lead
to the presence of the two prominent classes (corresponding to the two
loops of the figure-8).

\section*{Acknowledgement}
This work has been partly supported by a grant from the French ANR,
ANR-09-BLAN-0331-01, NSF grants FODAVA 0808515, CCF 1011228, and
NSF/NIH grant 0900700.


%




\bibliographystyle{plain}
\bibliography{witnessed}

\appendix


\subsection*{Recovering Betti numbers}
Letting $K$ be the support of a measure $\mu$, and $P$ a point sample drawn from
a distribution $\nu$ approximating $\mu$, 
we denote with $K^r$ and $P^r$ the sublevel sets $\dd_K(-\infty, r]$ 
and $\dw_{P,k}(-\infty, r]$ of the distance to $K$ and the witnessed
$k$-distance to the uniform measure on $P$, respectively.
With $e(m_0) = 54 m_0^{-1/2} \sigma + 24 m_0^{1/\ell}\alpha_\mu^{-1/\ell}$,
we have the following sequence of inclusions:
$$ K^0 \subseteq P^{e(m_0)} \subseteq K^{2e(m_0)} \subseteq
    P^{3e(m_0)} \subseteq K^{4e(m_0)}. $$ 
Assuming $K$ has a weak feature size $R$, and $e(m_0) < R/4$,
function $\dd_K$ has no critical values in the range
    $(0, R) \supseteq (0, 4e(m_0))$, and therefore the rank of the
    image on the homology induced by inclusion $\Hgr(P^{e(m_0)}) \to
    \Hgr(P^{3e(m_0)})$ is equal to the Betti numbers of the set $K$.

\bigskip
\bigskip

\end{document}